\documentclass[PP]{AEA}
\usepackage[cmbold]{}
\usepackage{amsmath}
\usepackage{amssymb}
\usepackage{bbm}
\usepackage{float}
\usepackage{ntheorem}

\newtheorem{claim}{Claim}

\newtheorem{definition}{Definition}
\newtheorem{proposition}{Proposition}

\theoremstyle{plain} 
\theorembodyfont{\upshape}

\newtheorem{example}{Example}

\makeatother

\newcommand{\indep}{\perp \!\!\! \perp}

\usepackage{natbib}

\draftSpacing{1.5}

\begin{document}

\title{
Algorithmic Fairness and Social Welfare}
\shortTitle{Short title for running head}
\author{Annie Liang and Jay Lu\thanks{Liang: Northwestern, 2211 Campus Drive, Evanston IL 60208, annie.liang@northwestern.edu. 
Lu: UCLA, 315 Portola Plaza
Los Angeles, CA 90095, jay@econ.ucla.edu. We thank Andrei Iakovlev for helpful comments.}}\pubMonth{Month}
\date{\today}
\pubYear{Year}
\pubVolume{Vol}
\pubIssue{Issue}
\maketitle

Algorithms are increasingly used to guide high-stakes decisions about individuals (for example, which patients to treat, or what kind of loan to offer a prospective borrower). Consequently, substantial interest has developed around defining and measuring the ``fairness'' of these algorithms \citep{EthicalAlg,barocas-hardt-narayanan}.

These definitions of fair algorithms share two features: First, they prioritize the role of a pre-defined  group identity (e.g., race or gender) by focusing on how the algorithm's impact differs systematically across groups.\footnote{Exceptions include, among others, \citet{DworkHardtPitassiReingoldZemel}'s comparison of outcomes across similar individuals.} Second, they are statistical in nature; for example, comparing false positive rates, or assessing whether group identity is independent of the decision (both viewed as random variables).

These notions are facially distinct from a social welfare approach to fairness, in particular one based on
 ``veil of ignorance'' thought experiments in which individuals choose how to structure society prior to the  realization of their social identity. 
In this paper, we seek to understand and organize the relationship between these different approaches to fairness. Can the optimization criteria proposed in the algorithmic fairness literature also be motivated as the choices of someone from behind the veil of ignorance? If not, what properties distinguish either approach to fairness?

We review the algorithmic fairness approach in Section \ref{sec:CS}, and the social welfare approach in Section \ref{sec:Veil}. We show in Section \ref{sec:Compare} that these approaches are fundamentally different. We conclude by proposing a framework that nests both approaches  in Section \ref{sec:Generalization}.

\section{Setup}
There is a population of subjects each described by three variables: a covariate vector $X \in \mathcal{X}$, a type $Y \in \mathcal{Y}$ (i.e., what the algorithm would like to predict), and a group identity $G \in \mathcal{G}$.  For simplicity of exposition, we assume the sets $\mathcal{X}$, $\mathcal{Y}$, and $\mathcal{G}$ are finite. The population distribution relating covariate vectors, types, and group identities is denoted by  $\mu \in \Delta(\mathcal{X} \times \mathcal{Y} \times \mathcal{G})$.

Each subject receives a decision from the finite set $\mathcal{D}$ as the (randomized) output of an algorithm
$a: \mathcal{X} \rightarrow \Delta(\mathcal{D}).$
Together with $\mu$, each algorithm $a$ pins down a full joint distribution  $P \in \Delta(\Omega)$ over the entire state space $\Omega = \mathcal{X} \times \mathcal{Y} \times \mathcal{G} \times \mathcal{D}$. We define $\mathcal{P}_\mu \subseteq \Delta(\Omega)$ to be the set of feasible distributions that can be implemented by some algorithm $a$ given the population distribution $\mu$. A Designer chooses from the set of feasible distributions $\mathcal{P}_\mu$.\footnote{This is equivalent to choosing among algorithms but more convenient for the subsequent exposition.}

\section{Fair Algorithms} \label{sec:CS}

In a popular approach in the algorithmic fairness literature, the Designer  maximizes a measure of the accuracy of the algorithm subject to a fairness constraint. The primitives of this approach are an accuracy measure $v: \mathcal{D} \times \mathcal{Y} \rightarrow \mathbb{R}$, and a fairness constraint (which we elaborate on below). The Designer solves the optimization problem
\begin{align} \label{eq:co} \sup_{P \in \mathcal{P}_\mu}& \mathbb{E}_P[v(D,Y)] 
 \quad \mbox{ s.t. } (*) 
\end{align}
where leading choices of the fairness constraint $(*)$ include the following (see e.g., \citet{Mitchelletal} for a reference):

\begin{definition}[Equalized Odds] \label{Def:EO} $D \indep G \mid Y$.

\end{definition}

\begin{definition}[Equality of False Negatives] $D \indep G \mid Y=1$.
\end{definition}

\begin{definition}[Equality of False Positives] $D \indep G \mid Y=0$.
\end{definition}

\begin{definition}[Statistical Parity] \label{Def:SP} $D \indep G$. 
\end{definition}

Algorithms that satisfy these fairness constraints are sometimes viewed as ``non-discriminatory.'' Under this interpretation, the Designer first commits to choosing a non-discriminatory algorithm and subsequently maximizes accuracy among them. The prioritization of a non-discriminatory algorithm  may reflect the designer's intrinsic preferences or external constraints (e.g., imposed by policy, law, or public opinion); the optimization problem in (\ref{eq:co}) is consistent with either interpretation.

\section{Social Welfare Approach} \label{sec:Veil}

The approach discussed above is qualitatively different from a long tradition in moral philosophy and economics, which instead measures social welfare by aggregating utilities across individuals in society. Here, fairness considerations stem from contemplating how  an individual would choose to structure society prior to the realization of the individual's own identity \citep{Harsanyi1953,Harsanyi1955,Rawls}. From behind this ``veil of ignorance,''  a preference for fairness can be micro-founded as risk aversion over the realization of identity. We review a formalization of this approach. 

The Designer forms preferences by considering a hypothetical individual whose identity has yet to be realized. The possible identities are indexed by $i \in \mathcal{I}$ (each realized with probability $p_i$), with $u_i$ denoting the eventual utility under identity $i$. The individual's ex-ante payoffs, aggregating over the possible identities, are given by
$\sum_{i \in \mathcal{I}} p_i \phi(u_i)$
for some concave and strictly increasing function $\phi: \mathbb{R} \rightarrow \mathbb{R}$.\footnote{For an axiomatization and further discussion of this representation, see \cite{GrantKajiiPolakSafra}.} 

When $\phi$ is linear (corresponding to risk-neutrality), this representation returns utilitarianism, à la Harsanyi. Strict concavity of $\phi$ corresponds to risk aversion over which identity one assumes following the lifting of the veil of ignorance, and characterizes a preference for fairness. For example, if $\phi(u) = \sqrt{u}$ and there are two equally-likely individuals, then the ``more fair'' utility vector $(u_1,u_2) = (2,2)$ is preferred over the ``less fair'' utility vector $(u_1,u_2)=(1,3)$. An infinite risk-aversion limit returns Rawls' preferred perspective, in which social outcomes are  evaluated based on the most disadvantaged individual, i.e. $\min_{i \in \mathcal{I}} u_i$.\footnote{Let $\phi(u)=-u^{-\gamma}$ and consider preferences in the limit as $\gamma\rightarrow\infty$.}

This framework can be adapted to accommodate a group-based notion of fairness. To do this,  we extend the thought experiment and suppose our Designer adopts the perspective of an individual behind the veil who anticipates a \emph{compound} lottery: First, group identity is realized, and then a second lottery determines the remaining characteristics. If the individual  has different risk preferences over the two-stages of the compound lottery, the Designer will solve the following maximization problem
\begin{align} \label{eq:sw} \sup_{P \in \mathcal{P}_\mu}\sum_{g}p_{g}\phi\left(U\left(P_{g}\right)\right) 
\end{align}
where $P_g$ denotes the conditional distribution of $P$ given $G=g$ and $U(P_g) = \mathbb{E}_{P_g}[u(D,Y)]$ is the expected utility of group $g$.\footnote{This setup and representation resemble models of second-order expected utility; see \cite{KrepsPorteus78}. We could also allow $U(\cdot)$ to capture different risk preferences over the remaining covariates, and do not expect our results to change.} Concavity of $\phi$  now  reflects aversion to group-identity risk.

\section{Comparing the Two Approaches}
\label{sec:Compare}

We first demonstrate in a simple example that the two approaches can be incompatible.

\begin{example}
\label{eg:1}
The population is divided equally into two groups $G \in \{0,1\}$. The type $Y \in \{0,1\}$ is an indicator for whether a patient needs treatment, and the decision $D \in \{0,1\}$ is an indicator for whether the patient is treated. The only observable covariate is $X=G$, so the decision $D$ is the (random) outcome of an algorithm $a: \mathcal{G} \rightarrow \Delta(\{0,1\})$.

There is a $\delta \in (1/2,1)$ such that the population distribution $\mu$ satisfies
\begin{equation} \label{eq:mu}
    \mu(Y=j \mid G=j) = \delta
\end{equation}
for both $j\in \{0,1\}$. That is, most members of group $G=1$ need treatment, while most members of group $G=0$ do not need treatment.

Consider two Designers. The first solves the constrained optimization problem in (\ref{eq:co}) given the Equalized Odds constraint in Definition \ref{Def:EO} ($D \indep G \mid Y$). We leave the accuracy measure $v$ unspecified.

The second is our social welfare Designer who solves the maximization problem in  (\ref{eq:sw}), endowed with utility function $u(d,y) = \mathbbm{1}(d=y)$ and any concave and strictly increasing $\phi$.

\begin{claim} \label{claim:Example1} In Example \ref{eg:1}, any solution to the constrained optimization problem (\ref{eq:co}) is not a solution to the social welfare problem (\ref{eq:sw}), and vice versa.
\end{claim}

\begin{proof} Since the only covariate on which the algorithm conditions is group identity, each algorithm can be identified with a pair $(q_0,q_1)   \in [0,1]^2$ where $q_g = P_g(D=1)$. Since every pair is feasible, it is possible for the social welfare Designer to optimize for the expected utilities of the two groups separately. His payoff is maximized by setting $q_1=1$ and $q_0=0$, yielding a  payoff of $\frac12\phi(\delta)+\frac12\phi(\delta)=\phi(\delta)$. This algorithm clearly fails Equalized Odds. In contrast,  an  algorithm satisfies Equalized Odds if and only if $q_0=q_1 := q$. 
\end{proof}

\medskip

In this example, the two Designers are in conflict. The social welfare Designer prefers to treat everyone in group $G=1$ and no one in group $G=0$, but this algorithm fails Equalized Odds and  so is unacceptable to the constrained optimization Designer. 

In the other direction, any algorithm that satisfies Equalized Odds leads to a lower payoff for the social welfare Designer. How large is the loss? We showed in the proof of Claim \ref{claim:Example1} that the social welfare Designer's preferred algorithm gives him a payoff of $\phi(\delta)$. In contrast, any algorithm that satisfies Equalized Odds gives the social welfare Designer a payoff of    \begin{align} \label{eq:UtilityCO}
   \frac12 & \phi\left(q\delta+(1-q)(1-\delta)\right) \\
   & \quad + \frac12 \phi\left((1-q)\delta+q(1-\delta)\right). \nonumber
\end{align}
Since $\phi$ is concave, Jensen's inequality implies that (\ref{eq:UtilityCO}) is less than
  \[  \phi\left(\frac12 \delta + \frac12 (1-\delta)\right)  = \phi\left(\frac12\right)   < \phi(\delta)
    \]
where the inequality follows from our assumptions that $\delta >1/2$ and that $\phi$ is strictly increasing.

Thus, not only is the social welfare Designer's payoff strictly lower under the constrained optimization algorithm,
but interestingly the disparity in payoffs becomes  larger the more the social welfare Designer cares about fairness (i.e., the more concave $\phi$ is). So the normative content of the two approaches about fairness are qualitatively different: Not only do the Designers have different optimal algorithms, but their fairness objectives push them in contrasting directions.
 
 We observe finally that the disagreement between designers cannot be resolved by  relaxing the fairness constraint in the constrained optimization problem. To see this, observe that the social welfare Designer's preferred algorithm perfectly correlates $D$ and $G$. Thus it not only fails Equalized Odds, but would fail any moderate relaxation of Equalized Odds as well. For example, if Equalized Odds is relaxed  to the condition that 
\begin{align} \label{eq:EOrelax} \vert \mathbb{E}[D  &\mid Y=y, G=1] \\
& - \mathbb{E}[D \mid Y=y,G=0] \vert \leq \varepsilon. \nonumber \end{align}
for each $y$, then the statement of the claim holds for all  $\varepsilon<1$.
\end{example}
\bigskip

We now show that these findings are not special to our particular specifications of the two Designers (e.g., the choice of utility in the social welfare approach or the consideration of Equalized Odds as the fairness constraint) but reflect a more general tension between the approaches. Towards this, note that the Designer's choices can be modeled as a \emph{choice rule} $c: \Delta(\mathcal{X},\mathcal{Y},\mathcal{G} ) \rightarrow \Delta(\Omega)$ specifying a feasible distribution $c(\mu)\in \mathcal{P}_\mu$  for every population distribution $\mu$. Let  $\mathcal{C}_{co}$ be the set of choice rules for all constrained optimization Designers, as we vary over the fairness constraints $(*)$ from Definitions \ref{Def:EO}-\ref{Def:SP} and all accuracy measures $v$  (with an arbitrary tie-breaking rule).\footnote{If no feasible distribution satisfies the fairness constraint, then pick any feasible distribution.} Let $\mathcal{C}_{sw}$ be the set of choice rules for all social welfare Designers, as we vary over all concave and strictly increasing functions $\phi$ and all \emph{nontrivial} utility functions $u$  (defined  below), again with any tie-breaking rule.\footnote{The subsequent Proposition 1 holds without the concavity restriction but we maintain it for ease of comparability with the algorithmic fairness approach.}

\begin{definition} A utility function $u:\mathcal{D} \times \mathcal{Y} \rightarrow \mathbb{R}$ is \emph{nontrivial} if not all types rank decisions the same, i.e. $u(d,y) > u(d',y)$ and $u(d,y')<u(d',y')$ for some  $d,d',y,y'$.
\end{definition}

The following result says that the constrained optimization and the social welfare approaches yield distinct choice rules.

\begin{proposition} \label{prop:Main}
Any constrained optimization choice rule is not a social welfare choice rule, and vice versa, i.e. $\mathcal{C}_{co} \cap \mathcal{C}_{sw} = \emptyset$
\end{proposition}

\begin{proof} It is sufficient to show that for every $v$, $(*)$, nontrivial $u$, and strictly increasing and  concave $\phi$, there exists a $\mu$ such that $c(\mu)$ is different under the two approaches. We will show this by construction, adapting Example 1. We prove the result assuming binary $Y$, $G$, and $D$, although the argument can be extended otherwise. Suppose that $X=G$ and define $q_g = P_g(D=1)$, observing that each algorithm is again identified with the pair $(q_0,q_1)$.

By assumption that $u$ is nontrivial, there exist $y_0,y_1 \in \mathcal{Y}$ such that $u(1,y_1) > u(0,y_1)$ but $u(1,y_0) < u(0,y_0)$. Consider a population distribution $\mu$ where there is a $\delta$ such that
\begin{align*}
    \mu(Y=y_j &\mid G=j)=\delta
\end{align*}
for each $j\in \{0,1\}$. Then (\ref{eq:sw}) reduces to choosing $(q_0,q_1) \in [0,1]^2$ to maximize
\begin{align*}
     & \frac12 \phi( (1-q_1)\delta  u(0,y_1) +q_1(1-\delta)  u(1,y_0) \\
     & \,\, \,\, \,
     + q_1\delta u(1,y_1)+ (1-q_1)(1-\delta)  u(0,y_0)) + \\ 
     & \frac12 \phi( (1-q_0)(1-\delta)  u(0,y_1) +q_0\delta  u(1,y_0) \\
     & \,\, \,\, \,
     + q_0(1-\delta) u(1,y_1)+ (1-q_0)\delta  u(0,y_0))
\end{align*}
We can choose $\delta$ large enough such that the unique solution is $(q_0,q_1)=(0,1)$. But each of the fairness constraints in Definitions \ref{Def:EO}-\ref{Def:SP} is satisfied if and only if $q_0=q_1$, so we have the desired result. 
\end{proof}

\medskip

The result does not say that social welfare Designers and constrained optimization Designers never agree.\footnote{For example, if (\ref{eq:mu}) in Example 1 were replaced by 
$\mu(Y=1 \mid G=g) = \delta$, then the algorithm that treated everyone in the population would maximize the social welfare Designer's payoff, and also satisfy Equalized Odds.} Instead, for any  social welfare Designer and any  constrained optimization Designer, there is \emph{some} population (i.e. $\mu$) for which the Designers' preferred distributions diverge. Since to our knowledge neither view of fairness is predicated on specific assumptions on the population, we interpret the approaches as intended ``universally,'' and Proposition \ref{prop:Main} as thus establishing that the two approaches are distinct.

Intuitively, the result follows  because the two Designers view fairness differently. The constrained optimization Designer requires conditional independence of certain variables given others, which is a \textit{statistical} property across the two groups. The social welfare Designer is instead concerned with what the joint distribution of all the variables implies for the \emph{utilities} across the two groups. When not all distributions on the state space are feasible, these two approaches can be in conflict with one another. 

The challenge of micro-founding the constrained optimization approach using a social welfare narrative does not necessarily imply that those approaches are misguided, but does imply that they require novel justifications. We leave exploration of such questions for future work.

\section{A More General Framework} \label{sec:Generalization}

In this final section, we propose a general framework that nests the  two approaches. First note the constrained optimization problem specified in (\ref{eq:co}) can be reformulated as
\begin{align}   \sup_{P \in \mathcal{P}_\mu}& \mathbb{E}_P[v(D,Y)]- h(P_{g_1},\dots,P_{g_n})  \nonumber
\end{align}
where $h$ is an indicator function that returns infinite cost if the fairness constraint $(*)$  is violated. For example, the Equalized Odds constraint (Definition \ref{Def:EO}) can be expressed as 
\begin{align}   {P_{g_1}(D|Y=y)}=\dots ={P_{g_n}(D|Y=y)}  \nonumber
\end{align}
for all $y$. Similar expressions can be written for the other constraints in Definitions \ref{Def:EO}-\ref{Def:SP}.\footnote{The same holds for any relaxed versions of these constraints, e.g. condition (\ref{eq:EOrelax}).}

Extending this representation, we consider the following maximization problem
\begin{align*}   \sup_{P \in \mathcal{P}_\mu}& \phi(\mathbb{E}_P[v(D,Y)])- h(P)  
\end{align*}
where $\phi(\cdot)$ is a strictly increasing function and $h\left(\cdot\right)$ is a measure of unfairness that satisfies
$h\left(P\right)\geq0$ with $h\left(P\right)=0$ if $P_{g}$ is the
same for all $g$.  In addition to nesting the constrained optimization problem, this representation returns the social welfare formulation when $v=u$, $\phi$ is concave and
\[
h\left(P\right)=\phi\left(\sum_{g}p_{g}U\left(P_{g}\right)\right)-\sum_{g}p_{g}\phi\left(U\left(P_{g}\right)\right)
\]
This unfairness measure $h$ can be interpreted as comparing the  utility of an individual whose identity is known to be the expected group identity, versus the expected utility of an individual with a random group identity. By Jensen's inequality, $h(P)$ is always positive, and its value is larger when the distribution over utilities (induced by different realizations of the group identity) is more dispersed. Alternative choices for the unfairness measure $h$ may generate novel preferences for fair algorithms, and would be interesting avenues for future research.

\end{document}